\documentclass[sigconf,nonacm,screen,natbib=false]{acmart}
\usepackage[utf8]{inputenc}
\usepackage[T1]{fontenc}

\usepackage{microtype}

\usepackage{url}
\usepackage{hyperref}

\usepackage{subcaption}

\RequirePackage[
  datamodel=acmdatamodel,
  url=false,
  maxcitenames=2,
  maxbibnames=9,
  giveninits,
  ]{biblatex}

\usepackage{multirow}
\usepackage{booktabs}

\usepackage[noline,noend,algo2e,linesnumbered]{algorithm2e} 

\RequirePackage{multicol}
\RequirePackage{amsmath, amsthm, mathtools}
\usepackage{amsfonts}

\input{mathdefs}

\newtheorem{definition}{Definition}[section]
\newtheorem{lemma}[definition]{Lemma}
\newtheorem{proposition}[definition]{Proposition}

\begin{document}

\SetInd{0.5em}{0.8em}

\SetKwFor{For}{for}{:}{}
\SetKwFor{If}{if}{:}{}
\SetKwFor{While}{while}{:}{}


\DontPrintSemicolon
\SetKwBlock{parameters}{parameters}{}
\SetKwBlock{constants}{constants}{}
\SetKwBlock{types}{types}{}
\SetKwBlock{state}{state}{}
\SetKwBlock{periodically}{periodically}{}
\SetKwBlock{on}{on}{}
\SetKwBlock{procedure}{procedure}{}
\SetKwBlock{function}{fn}{}
\SetKw{return}{return}
\SetKw{kwbreak}{break}
\SetKw{record}{record}
\SetKw{kwin}{in}
\SetKwFor{kwloop}{loop}{}{endloop}

\auxfun{causalSend}
\auxfun{bit}
\auxfun{receive}
\auxfun{send}
\auxfun{deliver}
\auxfun{trySend}
\auxfun{requestSlots}
\auxfun{add}
\auxfun{remove}
\auxfun{append}
\auxfun{enqueue}
\auxfun{dequeue}
\auxfun{peek}
\auxfun{first}
\auxfun{next}
\auxfun{size}
\auxfun{updatePositions}
\auxfun{sendInterval}
\newcommand{\M}{\mathds{M}}
\newcommand{\Msg}{\text{Msg}}
\newcommand{\Per}{\text{Per}}
\newcommand{\Rcv}{\text{Rcv}}

\title{Space-Optimal, Computation-Optimal, Topology-Agnostic,
Throughput-Scalable Causal Delivery through Hybrid Buffering}
\author{Paulo Sérgio Almeida}
\affiliation{%
  \institution{INESC TEC \& University of Minho}
  \country{Portugal}
}

\begin{abstract}
Message delivery respecting causal ordering (causal delivery) is one of the
most classic and widely useful abstraction for inter-process communication
in a distributed system.
Most approaches tag messages with causality information and buffer them at
the receiver until they can be safely delivered. Except for specific
approaches that exploit communication topology, therefore not generally
applicable, they incur a metadata overhead which is prohibitive for a
large number of processes.
Much less used are the approaches that enforce causal order by buffering
messages at the sender, until it is safe to release them to the network, as
the classic algorithm has too many drawbacks.
In this paper, first we discuss the limitations of sender-only buffering
approaches and introduce the Sender Permission to Send (SPS) enforcement
strategy, showing that SPS + FIFO implies Causal.
We analyze a recent sender-buffering algorithm, Cykas, which follows
SPS + FIFO, albeit very conservatively, pointing out throughput scalability
and liveness issues.
Then, we introduce a novel SPS + FIFO based algorithm, which adopts a new
hybrid approach: enforcing causality by combining sender-buffering to
enforce SPS and receiver-buffering to enforce FIFO. The algorithm overcomes
limitations of sender-only buffering, and achieves effectively constant
metadata size per message. By a careful choice of data-structures, the
algorithm is also computationally-optimal, with amortized effectively constant
processing overhead. As far as we know, there is no other topology-agnostic
causal delivery algorithm with these properties.
\end{abstract}

\maketitle

\section{Introduction}

Causal order message delivery (causal delivery, for short) is enforcing a
message delivery order which is consistent with potential causality, i.e., the
\emph{happened before} relation as defined
by~\textcite{DBLP:journals/cacm/Lamport78}. This allows defining
application-agnostic messaging algorithms that will respect any actual causal
dependencies resulting from application semantics.

Causal delivery has long been recognized~\cite{DBLP:journals/ipl/RaynalST91}
as an important concept with many applications. It stands out as the
sweet-spot, stronger than FIFO ordering, but not over-constraining, as total
ordering. Being based on \emph{happened before}, it is the strongest ordering
which is spatially scalable, given message propagation limited by light-speed:
unlike for total ordering, message delivery does not need to be delayed by
the occurrence of concurrent messages on the other side of the world. Causal
broadcast (sending to a set of processes) has an important role, to
implement highly available systems satisfying \emph{Causal
Consistency}~\cite{DBLP:journals/dc/AhamadNBKH95}, broadly the strongest
partition-tolerant consistency model.

Causal delivery was proposed~\cite{DBLP:journals/tocs/BirmanJ87} for the
special case of causal broadcast, in which each process broadcasts
messages to all $n$ processes in a group. In the initial proposal, causally
preceding messages were piggybacked with each sent message, which allows
immediate delivery, but is very inefficient in terms of communication
bandwidth, and was soon abandoned.
The general cases of unicast or multicast (where a set of destinations can be
specified at each send) render piggybacking predecessor messages unthinkable,
being algorithmically more demanding than broadcast.

Most early proposals, whether for
broadcast~\cite{DBLP:journals/tocs/PetersonBS89}, or
multicast~\cite{DBLP:conf/wdag/SchiperES89,DBLP:journals/ipl/RaynalST91}
enforce causal ordering by delaying delivery at the receiver. They keep
received messages in a buffer, until they are deemed safe to be delivered to
the process, i.e. until all causally preceding messages which are still
missing are received and delivered. The main advantage is being able to
minimize delivery latency. The big drawback is the need to tag messages with
metadata describing the causal past, to be able to identify which
messages are still missing and must be waited for.

Even for the broadcast case, metadata already requires $O(n)$ integers for the
canonical approach. It can be improved in practice by only sending
information about the immediate causal
predecessors~\cite{DBLP:journals/tocs/PetersonBS89}, but it is still
$O(n)$ in the worst case.

For the general case of unicast or multicast messaging, the canonical
Raynal-Schiper-Toueg (RST) causal delivery
algorithm~\cite{DBLP:journals/ipl/RaynalST91} has $n^2$ integers of message
metadata.
The algorithm by \textcite{DBLP:journals/jpdc/PrakashRS97} reduces message
metadata by only sending the immediate causal predecessors per destination, but
still uses an $n^2$ sized matrix in the process state.
Moreover, it has been shown~\cite{DBLP:journals/dc/KshemkalyaniS98} that
topology-agnostic unicast/multicast algorithms with receiver-buffering which
achieve optimal delivery latency need $\Omega(n^2)$ identifiers of message
metadata, with a metadata-optimal algorithm presented in the same paper. Even
if in practice only a fraction of the canonical $n^2$ integers are needed for
the optimal algorithm, from 5\% to 35\% in the many scenarios presented
by~\textcite{DBLP:journals/tpds/ChandraGK04}, having already 5\% of $100
\times 100$, i.e., 500 integers metadata per message, for a 100 processes
scenario, is arguably too much. For large scale scenarios with thousands of
processes, even the optimal receiver-buffering algorithm is completely
prohibitive.

Algorithms for large scale systems are possible, overcoming the need for
metadata in messages and the need to buffer and reorder messages at the
receiver, by making sure that messages are received already respecting
causal order.

One way to do that is by controlling the communication topology, by building
an overlay network that is used to carry messages.
The more obvious case can be summarized by
``FIFO + Tree -> Causal'', i.e., if messages are propagated along a
spanning tree, with each link being FIFO, then causal ordering ensues,
trivially. In this case, there is no need for message metadata or buffering
messages at the receiver before delivery. This strategy has been around for
a long time, at least since being presented and proved correct by
\textcite{DBLP:conf/hicss/BonaR96}, but has remained relatively little known,
having been rediscovered independently several times, e.g., for
geo-replication, in \textsc{Saturn}~\cite{DBLP:conf/eurosys/BravoRR17}, or by
\textcite{DBLP:conf/agere/BlessingCD17}, motivated by developing a runtime
for a distributed actor language.

But a tree is not suitable for large scale: the need for tree
reconfiguration upon node or link failures, bound to happen at scale, implies
complexity and delays; links near the network center become a bottleneck; long
propagation paths, with forwarding at each node by the messaging middleware,
impose substantial latency, compared with the routing at the network layer by
the optimal path.
\textsc{Saturn} itself, while for large geographic scales, is used for a small
number of nodes (datacenters plus some extra serializer nodes), and the tree
is used only for metadata propagation, with the bulk data being propagated
directly node-to-node (using network layer routing). While this strategy may be
suitable for large messages, where the bulk data takes some time to propagate,
that will not be the case for small messages. A tree topology will not be
suitable for large scale systems with thousands of processes and possibly
billions of small messages per second of aggregated throughput.

Exploiting communication topology towards reducing message metadata and
process state had already been done~\cite{DBLP:conf/podc/MeldalSV91} towards
determining potential causality between messages for the same destination. But
for causal delivery, being able to causally compare two given messages is not
enough: it needs the knowledge of whether there are still missing messages.

Exploiting topology for causal delivery (more generally than the tree case)
was achieved by Causal Separators ~\cite{DBLP:conf/icdcs/RodriguesV95}:
essentially, a set $S$ of nodes which separates the rest of the network in two
disjoint sets $A$ and $B$, where any message from a node in $A$ to a
node in $B$ must pass through a node in $S$. This can be exploited, e.g., when
all messages from processes internal to a datacenter to external processes must go
through one or more gateways, which serve as causal separators. But it does
not solve the problem of how to perform communication between possibly thousands of
processes within the datacenter itself, if internal communication topology is
arbitrary.

An alternative approach to ensure causal delivery, much less explored, is
buffering messages at the sender, delaying their release to the network (the
network-send) until it is safe to do so, ensuring that messages will arrive
respecting causal ordering, which avoids the need to buffer and reorder at the
receiver. A simple proposal by~\textcite{DBLP:conf/dagstuhl/MatternF94} (MF)
follows this approach, using a FIFO sender buffer, so as to be non-blocking,
but allowing a single unacked message in transit. (We abbreviate
acknowledgment, acknowledged and unacknowledged as ack, acked, and unacked.)
This is very limiting, almost synchronous, with too many drawbacks to be used
in practice.

But a recent paper by \textcite{TongKuper2024} proposes an algorithm (Cykas:
Can you keep a secret?) which follows the approach of buffering at the sender,
while introducing a new interesting ingredient: decoupling the delivery of a
message $m$ from the ability of its receiver to release to the network
subsequent messages, with the permission to do so granted by the sender of
$m$.
This allows earlier message release than the MF algorithm and several messages
in transit, decreasing latency and allowing more concurrency.

But as we discuss in this paper, any purely sender-buffering approach
will not be suitable, namely to achieve \emph{scalable throughput}, not even
allowing message pipelining. Cykas itself is not optimal in the
sender-buffering design space, and manifests several problems which
we discuss, namely the possibility of starvation.

In this paper we introduce a new strategy to enforce causal delivery which defines
a class of algorithms in which Cykas fits as a very conservative (non-optimal)
instantiation. We then introduce a new topology-agnostic causal delivery
algorithm, which is an optimal instantiation for the new class.
It uses an hybrid approach of enforcing causal delivery by having both
sender-buffering (delaying message release to the network) and
receiver-buffering (to reorder received messages from the same sender,
allowing pipelining). It follows the same ingredient as Cykas, in decoupling
delivery from the ability to do subsequent network-sends, but in a different
way, overcoming some problems in Cykas (namely starvation). It achieves
optimal metadata complexity (effectively constant size message and process
metadata per message) while also being computationally optimal, having
amortized constant-time operations, through the use of carefully designed
data-structures. While delivery latency is obviously not optimal, it allows
scalable throughput, unlike Cykas. For large scale scenarios, if
latency is less important than throughput and the cost of processing
each message should be minimized, it is arguably the best algorithm. It may in
fact be the only scalable choice if no assumption can be made regarding
topology, or used as a component in a topology-aware approach (e.g., within
each datacenter, scaling to thousands of processes, complemented with a
tree-based approach like \textsc{Saturn}, for inter-datacenter communication).

The contributions of this paper are:

\begin{itemize}
  \item Introduce criteria for comparing causal delivery algorithms, in
    particular those that are not latency-optimal, and perform a comparison of
    representative algorithms according with the criteria.
  \item Introduce a class of causal delivery algorithms
    based on what we call \emph{Sender Permission to Send (SPS)} enforcement
    strategy, show that SPS + FIFO implies Causal, analyze what seems to
    be the only sender-buffering algorithm which follows SPS + FIFO, albeit
    very conservatively (Cykas), show it to be non-optimal (for the class),
    and discuss its liveness problems.
  \item Introduce a novel SPS + FIFO causal delivery algorithm with an
    hybrid-buffering approach, using both sender- and receiver-buffering,
    which is topology-agnostic, space-optimal and computation-optimal, while
    allowing scalable throughput.
\end{itemize}

\section{Criteria for causal delivery algorithms}

\subsection{Space overhead}

The main criteria used to compare causal delivery algorithms has been the
space overhead: the overhead of metadata used for both messages in transit
and in the process memory. This is due to classic receiver-buffering
algorithms being similar in most other aspects, being latency-optimal and with
no throughput or liveness issues. But for other approaches, either using
sender-buffering or hybrid buffering, in addition to space overhead, other
criteria become relevant to their comparison.

\subsection{Delivery latency}

Delivery latency has not been explicitly used in comparisons, as most causal
delivery algorithms are latency-optimal. Optimality in this criteria has been
implicitly assumed almost as a requisite, with no tradeoff considered.
Optimality naturally ensues for the majority of algorithms, that do immediate
sends and receiver-buffering, while keeping precise knowledge about causal
dependencies. For them, all delivery latency is the minimum possible (except
algorithms which naively piggyback other causally preceding messages with each
individual message, unreasonable in practice). Delivery latency coincides with
the unavoidable causal latency~\cite{Ward2007AlgorithmsFC}.
Almost no receiver-buffering algorithm imposes additional latency
--- non-causal latency -- one exception being one algorithm which does not
keep precise knowledge, used in
Newtop~\cite{DBLP:conf/icdcs/EzhilchelvanMS95}, as we discuss below.

Sender-buffering algorithms, such as MF or Cykas, impose non-causal latency.
But considering only the delivery latency of each message may not provide the
full picture for comparing algorithms, specially those that are not
latency-optimal. Another important aspect is throughput in the system.

\subsection{Throughput}

There may be scenarios where the aim is to achieve the maximum amount of work
done given some computational resources, measured in messages per second
handled, i.e., throughput, with latency being less relevant.
System-wide causal delivery may be desired anyway, to avoid strange phenomena
which could occur otherwise. Scenarios may be from the simple case of a
processing pipeline up to a large DAG of message-based microservices, where
messages start from some sources and are processed along several steps, with
a mix of dispatching to the appropriate service and load-balancing.
For the case of clients doing RPCs to a service, each client being limited by
latency, we want to be able to scale aggregated throughput, i.e., be able to serve
a number of clients which scales with computational resources, regardless of
network latency.

As we discuss below, sender-buffering algorithms not only lead to non-causal
latency but also are easily prone to fail in being scalable in what
concerns throughput. A good criteria is whether the throughput can scale
with computational power or if it is limited by network latency between
processes. A good algorithm should have \emph{spatially scalable throughput},
where throughput is not hindered by the spatial distance between nodes. The
algorithm presented in this paper allows it, unlike sender-only buffering algorithms.

A simple metric that may give a hint about the achievable throughput and
help compare different algorithms is the number of messages in transit allowed
(per sender). MF allows a single message in transit, Cykas allows at most
one message per destination in transit, while our algorithm allows several
messages per destination. (In addition to allowing earlier network-sends in
general.)

\subsubsection{Pipelining}

A special but important case is pipelining. To achieve high throughput
regardless of network latency, it is essential to be able to have several
messages in transit between each pair of processes. This implies that
algorithms for the general case of networks which are unreliable or do not
ensure FIFO will need receiver-buffering to allow pipelining. Therefore, even
if only aiming for algorithms that are space-optimal and with no concern for
delivery latency, algorithms using purely sender-buffering must be ruled out.
This is a strong motivation for a new class of hybrid buffering
algorithms -- having both sender- and receiver-buffering -- which we introduce
in this paper.

\subsection{Liveness}

The possibility of deadlocks for blocking-sends and synchronous
communication was discussed by~\textcite{DBLP:conf/dagstuhl/MatternF94},
namely the scenario where each process is waiting for the ack of the delivery of
its own message and does not accept the message from the other process. This
was one motivation for their introduction of sender-buffering. Less
obvious liveness issues may arise in blocking algorithms even using
asynchronous communication, such as discussed
by~\textcite{DBLP:journals/tc/AnastasiBG04} regarding the algorithm
by~\textcite{10.1109/12.737681}.

But even non-blocking algorithms using sender-buffering can have problems
regarding progress such as the occurrence of starvation in some scenarios.
This may not be immediately obvious when designing an algorithm, as we
illustrate below for Cykas. Our algorithm does not suffer from liveness
problems.

\subsection{Computation time}

Delivery algorithm computation time has been mostly ignored, assumed to be a
negligible overhead, compared with the actual processing of messages
themselves. But if considering the scalability of classic algorithms to
thousands of processes, computation time may become non-negligible.
Classic latency-optimal algorithms have $O(n^2)$ computation overhead, becoming
non-negligible at large scale, specially in scenarios involving a simple
handling task per message. Sender- or hybrid-buffering allow algorithms with
much better computation complexity.

\section{Current topology-agnostic approaches to causal delivery}

We now compare briefly the two main approaches to causal delivery --
receiver-buffering and the less explored sender-buffering -- with our novel
hybrid approach. We only consider strategies that meet the following criteria:
\begin{itemize}
  \item A causal-send only network-sends the message itself plus
    metadata. We leave out strategies that piggyback with each message
    other causally preceding messages~\cite{DBLP:journals/tocs/BirmanJ87}, too
    inefficient and requiring garbage collection of messages to avoid
    unbounded growth.
  \item They have non-blocking sends, i.e., where the causal-send operation
    returns immediately. Approaches having a blocking send, e.g., blocking
    until the previous one has been acknowledged as
    delivered~\cite{10.1109/12.737681}, will prevent the process from
    progressing for some time. Non-blocking sends let the process progress, to
    deliver messages and perform work. This requires either immediate
    network-sends (most approaches, that assume reliable networks) or having a
    send buffer to hold the message before a network-send can be safely
    performed.
\end{itemize}

The results from analyzing the most representative algorithms are summarized
in Table~\ref{tab:comparison}, which shows: the network assumptions regarding
faults and order; space overhead and computation time (in terms of number of
integers or identifiers, as usually presented for causal delivery algorithms,
not in terms of bit complexity); whether an algorithm is latency-optimal or
where non-causal latency is caused (sender or receiver) and whether it is
unbounded or its magnitude as a function of the RTTs (round-trip time) between
processes (sender of the message to its peers); whether an algorithm allows scalable
throughput and the limit of messages in transit per sender; whether liveness
is assured; and whether an algorithm requires extra control messages.

\begin{table*}[t!]
  \caption{Comparison of topology-agnostic approaches to causal delivery,
  for representative sender-buffering, receiver-buffering and hybrid
  approaches. Space complexity in terms of integers/identifiers (not bit
  complexity). Considering number of processes ($n$), indegree $(i)$ /
  outdegree $(o)$ of the process (i.e., number of senders to the process and
  receivers), messages buffered $(b)$, permits missing $(p)$, unacked messages
  $(u)$, and the RTTs from sender to its peers $(r)$.}
  \label{tab:comparison}
  \begin{center}
    \begin{tabular}{@{}rlcccccc@{}}
  \toprule
    & & \multicolumn3c{receiver-buffering} & \multicolumn2c{sender-buffering}
      & \multicolumn1c{hybrid} \\
   \cmidrule(l){3-5}
   \cmidrule(l){6-7}
   \cmidrule(l){8-8}
    & & RST~\cite{DBLP:journals/ipl/RaynalST91}
      & KS~\cite{DBLP:journals/dc/KshemkalyaniS98}
      & Newtop~\cite{DBLP:conf/icdcs/EzhilchelvanMS95}
      & MF~\cite{DBLP:conf/dagstuhl/MatternF94}
      & Cykas~\cite{TongKuper2024}
      & this paper \\
  \midrule
      \multirow2*{network}
      & failures & reliable & reliable & reliable & any & reliable & any \\
      & order & any & FIFO & FIFO & any & any & any \\
  \midrule
      \multirow2*{space}
      & message & $O(n^2)$ & $O(n^2)$ & $O(1)$ & $O(1)$ & $O(1)$ & $O(1)$ \\
      & process & $O(n^2 b)$ & $O(n^2 b)$ & $O(n + b)$ & $O(b)$ & $O(no + b)$
      & $O(i+o+b+p+u)$ \\
  \midrule
      \multirow2*{computation}
      & sender & $O(1)^1$ & $O(n^2)$ & $O(1)$ & $O(1)$ & $O(no)$ & $O(1)$ \\
      & receiver & $O(n^2)$ & $O(n^2)$ & $O(n)$ & $O(1)$ & $O(1)$ & $O(1)$ \\
  \midrule
      \multirow2*{latency}
      & optimal & yes & yes & no & no & no & no \\
      & \multirow2*{non-causal}& no & no & receiver & sender & sender & sender \\
      & & 0 & 0 & unbounded $^2$ & $O(b r)$ & unbounded & $O(r)$ \\
  \midrule
      \multirow2*{throughput}
      & scalable & yes & yes & yes & no & no & yes \\
      & in transit & unlimited & unlimited & unlimited & 1 & 1/receiver & unlimited \\
  \midrule
      \multicolumn2c{liveness}
      & yes & yes & no $^2$ & yes & no & yes \\
  \midrule
      \multicolumn2c{extra control messages}
      & no $^3$ & no $^3$ & no $^3$ & \textsc{ack} & \textsc{ack,yct} &
      \textsc{ack,permit} \\
      \bottomrule
      \multicolumn8l{\footnotesize{$^1$ Considering only
      the algorithm steps and ignoring the cost of serializing the matrix,
      which is $O(n^2)$.}} \\
      \multicolumn8l{\footnotesize{$^2$ To overcome it, the
      service sends null messages if a process becomes idle for some time.}} \\
      \multicolumn8l{\footnotesize{$^3$ These algorithms
      assume a reliable network. For unreliable
      networks ACKs will be needed at the transport layer.}}
    \end{tabular}
  \end{center}
\end{table*}

Because in a large scale system with thousands of processes and many different
services, each process will typically communicate with only a very small
fraction of the whole, it is relevant for scalability that space overhead is
not affected by the total number of processes $n$ but only by the direct
peers. For space overhead at a process we distinguish, for the more scalable
algorithms, the indegree $i$ (number of senders to the process) and
outdegree $o$ (number of receivers). Algorithms that propagate information
transitively (classic ones like RST and KS) will have a space overhead that
depends on the total number of processes, having poor scalability.
Our algorithm has constant space per message in transit, constant computation
time (amortized), and process space that depends only linearly on individual
parameters, never on products of several parameters nor on the total number of
processes, unlike the other algorithms, except MF (which does not allow
scalable throughput). The process space complexity for our algorithm includes
the number of already network-sent but still unacked messages $u$, which is
required for any algorithm that tolerates message loss.

\subsection{Receiver-buffering}

Receiver-buffering algorithms use a buffer to hold received messages until
they are deemed safe to be delivered, either due to knowledge about still
missing causal predecessors that must be delivered first, or due to
uncertainty from the lack of knowledge.

The canonical RST~\cite{DBLP:journals/ipl/RaynalST91} algorithm stores at each
process an $n \times n$ matrix $M$, which describes the process knowledge about
all (transitive) causal dependencies of subsequent messages. $M[i,j] =
c$, means that $c$ messages were sent from process $i$ to process $j$ in the
causal past. The matrix is piggybacked as metadata with each message, and kept
after the message is received, while being buffered before delivery. The space
cost is therefore $O(n^2)$ integers per message, and $O(n^2b)$ at the process
when $b$ messages are buffered. (The algorithm also keeps a vector tracking
the number of delivered messages per origin.) Algorithmic processing cost is
$O(n^2)$ at the receiver, and constant time at sender (but serialization of
the matrix will have $O(n^2)$ cost, even if it may be negligible in practice).
This algorithm is latency-optimal, with no non-causal latency imposed.

For the class of latency-optimal algorithms, the
KS~\cite{DBLP:journals/dc/KshemkalyaniS98} algorithm is metadata-optimal: it
sends in messages and keeps at processes the minimum metadata so as to ensure
causal delivery as soon as possible.  Described for a multicast API, it
essentially keeps and propagates in messages, for each still relevant message
in the causal past, the set of remaining destinations for which it is not
known whether the message has been delivered or will be delivered respecting
causal order. Information is discarded as it becomes redundant, such as a
remaining destination if there is a subsequent message to the same
destination, or an entry with an empty set of remaining destinations if there
is a subsequent message from the same sender. Space space complexity (both per
message in transit and per message buffered to be delivered is $O(n^2)$.  Even
though space overhead may be much less in
practice~\cite{DBLP:journals/tpds/ChandraGK04}, for common scenarios
resembling something like $O(n^k)$ with $k$ closer to 1 than 2, it is always
larger than $\Omega(n)$, as at least an entry per sender is kept.  It is
therefore unsuitable for large scale.

Newtop~\cite{DBLP:conf/icdcs/EzhilchelvanMS95} has a receiver-buffering
algorithm for causal broadcast which requires less metadata by making
conservative decisions regarding delivery being safe. It delivers later than
the optimal due to the lack of knowledge. It assumes reliable FIFO
communication, and uses Lamport clocks~\cite{DBLP:journals/cacm/Lamport78} to
tag messages and keeps a vector of $n$ integers, with the clock of the last
received message from each process. Received messages with clock up to the
minimum of the $n$ integers can be safely delivered, in clock order. This
strategy achieves in effect a total order broadcast compatible with causal
delivery (if ties are broken, e.g., using process ids). But it causes unbounded
non-causal latency, as no message can be delivered until a message is received
from every process, being unsuitable in practice, even for causal broadcast,
and not at all suitable for general unicast, where each process may only
exchange messages with a subset of all processes (and it is not known a priori
which ones).

\subsection{Sender-buffering}

By sender-buffering algorithms we mean those that buffer messages requested to
be causally sent until it is safe for them to be network-sent (released to the
network by an actual send operation of the transport layer).
Almost no such algorithms exist. (We do not classify an algorithm as
sender-buffering, in what concerns causal delivery, if the algorithm
assumes an unreliable network and has a buffer to hold messages, already
network-sent, for possible retransmission until acknowledged. Variants of such
algorithms assuming a reliable network would not have a buffer.)

The classic representative of sender-buffering is the
MF~\cite{DBLP:conf/dagstuhl/MatternF94} algorithm. It has a FIFO send buffer,
to which messages requested to be causally sent are enqueued. This allows
non-blocking causal-sends, which return immediately, allowing the process to
perform subsequent work, possibly to process messages subsequently delivered.
The algorithm network-sends a single message at a time, waiting for
the acknowledgment that it has been delivered before dequeuing and
network-sending the next message from the FIFO send buffer. At the receiver a
message can be delivered immediately.

In fact, the algorithm is presented for a blocking receive API, using a
receive FIFO buffer from which messages are dequeued. But this buffer plays no
role in deciding delivery order, and is not needed for an API where delivery
is a triggered event.  The same applies to any algorithm described with a
blocking receive API.  Therefore, we classify MF as a purely sender-buffering
algorithm, in what causal delivery is concerned.


The algorithm ensures causal delivery trivially, as all causal dependencies
are ensured to be delivered before any message causally in the future is
network-sent. It has optimal metadata overhead per message, in transit and in
memory, and no computation cost, but is not adequate to practical use, even
for small scale.
As it allows a single message per sender in transit, the throughput is limited
by network latency. It is of little use to have a non-blocking send when
messages will accumulate unboundedly at the send buffer if the causal-send
rate exceeds the inverse of the average RTT to the destination. In addition to
non-scalable throughput and the buffer possibly growing unboundedly, the
delivery latency of a message will grow with the product of the average RTT by
the number of messages already in the buffer when the causal-send is
requested.

The very poor scalability properties of the MF algorithm has possibly
made sender-buffering algorithms not much further pursued. But a recent
algorithm, Cykas~\cite{TongKuper2024}, introduces a novel ingredient.
decoupling the delivery of a message $m$ from the ability of its receiver to
release to the network subsequent messages, with the permission to do so
granted by the sender of $m$. (In MF, the ability to perform a network-send
only depends on previous sends having been acked, not on received messages.)
The approach in Cykas allows earlier network-sends than the MF algorithm and
several messages in transit, decreasing latency and allowing more concurrency.

But by not having a receive-buffer, Cykas allows a single message per
destination in transit (per sender). This means that concerning the
interaction between two processes (e.g., in a processing pipeline), Cykas has
the same limitations as MF, with throughput not scaling with processing power,
but being limited by network latency. The Cykas algorithm is also sub-optimal
concerning how it enforces an implicit invariant it depends on to ensure
causal delivery, being conservative in deciding when to allow network-sends,
due to insufficient knowledge, as we explain below.

\section{A new approach to achieve causal delivery}

We now introduce a strategy to enforce causal delivery which defines a new
class of algorithms that generalizes what Cykas does.
Cykas fits as a very conservative (non-optimal) instantiation of the strategy,
our basic algorithm a less conservative instantiation, and our optimized
algorithm makes it latency-optimal for this class of algorithms.

\subsection{Causal delivery and happens-before between application messages}

The Lamport happened before relation is defined over sends and receives of
messages in a distributed system. But for causal delivery we are interested
only in the potential propagation of information from and into the processes
using the messaging abstraction, and not the happened before involving
(network-)sends and receives of all protocol messages. In a hybrid-buffering
algorithm, in addition to events concerning control messages, we have four
kinds of events concerning application messages: causal-send (c), network-send
(s), receive (r), and deliver (d). Only causal-send and delivery of messages
are visible to the application; the others are internal protocol events.
Happens-before ($\rel{hb}$) between application messages $m_1$ and $m_2$ is
defined as the transitive closure of:

\begin{enumerate}
  \item $m_1 \rel{hb} m_2$ if d$_i(m_1)$ occurs before c$_i(m_2)$ at a
    process $i$ that sent $m_2$;
  \item $m_1 \rel{hb} m_2$ if c$_i(m_1)$ occurs before c$_i(m_2)$ at a process
    $i$ that sent both $m_1$ and $m_2$.
\end{enumerate}
A causal delivery algorithm ensures that for any process $i$ that delivers
two messages $m_1$ and $m_2$:
\[
  m_1 \rel{hb} m_2 \text{ implies that d}_i(m_1) \text{ occurs before d}_i(m_2).
\]

Happens-before between messages is illustrated in
Figure~\ref{fig:happens-before}, showing that it is the relative occurrence
between delivery and a subsequent causal-send that defines happens-before. The
receive events do not matter, as the information has not yet reached the
process. Network-sends also do not matter, as information flows from the
process when causal-sent, being irrelevant whether the message is still in
the send-buffer.

\begin{figure}
\begin{center}
  \includegraphics[scale=1.2]{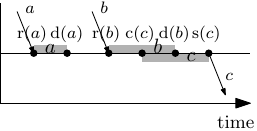}
\end{center}
  \caption{Happens-before between application messages. Some process receives (r), buffers, and then
  delivers (d) message $a$, and similarly for message $b$. The
  process causal-sends (c) message $c$, keeps it in the send-buffer,
  and then network-sends (s) it.
  We have $a \rel{hb} c$ but $b \not\rel{hb} c$,
  because, even if $b$ is received before, it is not
  delivered before $c$ is causal-sent. It does not matter that $b$ is
  delivered before $c$ is network-sent.
  }
  \label{fig:happens-before}
\end{figure}

\subsection{Sender permission to send enforcement strategy}

\begin{definition}[Sender Permission to Send (SPS) enforcement strategy]
A process $i$ only network-sends a message $m$ to a process $j$ when all
messages that happened before $m$, sent to any process other than
$i$, by any process $k \neq j$ that sent messages to $i$, have been delivered.
\end{definition}

If all processes enforce SPS (namely, those who sent messages to $i$),
transitive dependencies are enforced, as other messages that happened before
$m$, sent by any other process, must have been delivered.

\begin{proposition}
\label{prop:fifo-sps-causal}
An algorithm which enforces FIFO and SPS ensures causal delivery.
\end{proposition}

\begin{proof}
Given message $m$ sent by process $i$ to some process $j$, causal predecessors
of $m$ that must be delivered before $m$ at $j$ are either:
1) messages from $i$ itself, which will be delivered before $m$ due to FIFO, or
2) some message $m'$ sent to $j$ by some process $p_0$, that happened before
$m$ due to a causality path involving a set of processes
$\{p_0, p_{1}, \ldots, p_k = i\}$, for some $k > 0$, with each process $p_l$,
for $l < k$, sending a message to process $p_{l+1}$, and process $p_k = i$
sending $m$ to $j$.
As process $p_1$ enforces SPS, it will network-send the next message along
the path ($m$ if $k = 1$) only after $m'$ has been delivered.
This implies that all network-sends along the path, by any process
$p_l$, for $1 \leq l \leq k$, including process $i$ network-sending $m$, will
happen after $m'$ has been delivered.
\end{proof}

We now present a conservative variant of SPS that still allows algorithms with
scalable throughput,
implemented by our basic algorithm in Section~\ref{sec:basic-algorithm}.
An optimized algorithm which enforces SPS, allowing better delivery latency in
some cases is presented in Section~\ref{sec:optimizized-algorithm}.

\begin{definition}[Conservative Sender Permission to Send (CSPS) enforcement
strategy]
A process $i$ only network-sends a message $m$ when all messages that happened
before $m$, sent by any process that sent messages to $i$, have been delivered.
\end{definition}

\begin{proposition}
\label{prop:fifo-csps-causal}
An algorithm which enforces FIFO and CSPS ensures causal delivery.
\end{proposition}

\begin{proof}
By their definitions, an algorithm which enforces CSPS also enforces SPS, as
a process enforcing CSPS only network-sends a message when all conditions
required for SPS are met.
If it also enforces FIFO, by Proposition~\ref{prop:fifo-sps-causal} it
ensures causal delivery.
\end{proof}

It should be noted that SPS, CSPS and the propositions above are quite
general, allowing different algorithms, such as ours. Cykas uses an even more
conservative strategy than CSPS. To help understand SPS/CSPS-based
algorithms, consider some message $a$, and the stable property:
$D^a(t)$ = [messages sent before $a$ by the same sender have been delivered
before time $t$].
We can consider three different times related to $D^a(t)$:
\begin{description}
  \item[$O^a$]: the (smallest) time $t$ when $D^a(t)$ became true (known
    by an omniscient observer).
  \item[$S^a$]: the time when the sender of $a$ learns that $D^a(t)$ became
    true.
  \item[$R^a$]: the time when the receiver of $a$ learns that $D^a(t)$
    became true.
\end{description}

Figure~\ref{fig:dependencies} shows a run involving 6 processes, in which
process $i$ has causally-sent message $m$ (event c($m$)) after receiving
and delivering messages $a$ and $b$ (shown in solid bold) from processes $j$
and $k$, respectively. Network-sent messages are shown by solid lines, ack
control messages by dotted lines and what we call permission control messages
(or simply permits) by dashed lines.

\begin{figure}
\begin{center}
  \includegraphics[scale=1]{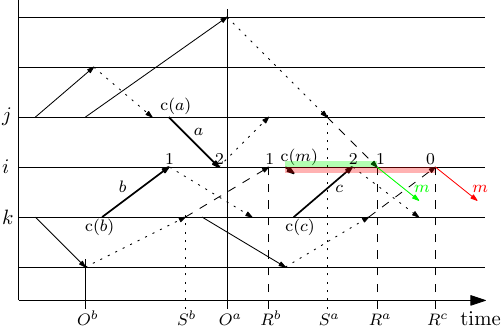}
\end{center}
  \caption{Spacetime diagram. Process $i$ causal-sends message $m$ after
  receiving and delivering message $a$ (from process $j$) and $b$ (from
  process $k$). Message $m$ is kept in the send-buffer until it can be
  network-sent, when $i$ knows that all causal predecessors sent by either $j$
  or $k$ have been delivered. Optimal time in the buffer shown in green; time
  taken by Cykas shown in red.}
  \label{fig:dependencies}
\end{figure}

The figure shows in the time axis, for messages $a$ and $b$, the times
corresponding to $O^a$, $S^a$, $R^a$, $O^b$, $S^b$, and $R^b$. Time $O^a$ is
when all messages sent by $j$ before $a$ have been delivered, time $S^a$
when all acks of those messages have been received by $j$, and $R^a$ when the
receiver of $a$, i.e., process $i$ has been informed (by a permission
control message). 

Consider the class of possible causal delivery algorithms based on CSPS,
using normal messages and these two kinds of control messages (acks and
permits). In an optimal algorithm, process $i$ should be able to keep $m$
in the send-buffer only until time $R^a$: the earliest time $t$ when $i$ could
learn that $D^a(t) \land D^b(t)$ becomes true, i.e., for all messages ($a$ and
$b$) delivered by $i$ before causally-sending $m$.  For the optimal algorithm,
$m$ should be in the send-buffer as depicted by the green rectangle (above the
timeline for process $i$).

\subsection{Analysis of the Cykas algorithm}

Cykas is a sender-buffering algorithm where nodes can immediately deliver
messages received but are restricted from network-sending causally subsequent
messages until given permission by their sender. Essentially, Cykas works
according to the following rules:

\begin{enumerate}
  \item There is a FIFO send-buffer where messages requested to be
    causally-sent are enqueued, no receive-buffer, and messages can be
    delivered immediately, upon being received.
  \item For each destination, a message is network-sent only after any
    previous message to the same destination has been acked.
  \item Messages are network-sent either as a normal-send (if all previously
    sent messages have already been acked) or as an eager-send, otherwise.
  \item At any time, a process is either in normal mode or secret mode, having
    an integer state variable MODE. Greater than zero means secret mode.
  \item Receiving an eager-send increments MODE, transitioning to, or
    remaining in secret mode.
  \item Messages causally-sent remain buffered while the process is in secret
    mode. In normal mode they can be network-sent (if rule 2 is also
    satisfied).
  \item If a process eager-sends a message $m$, it will subsequently send a
    YCT (you-can-tell) control message to the same receiver when all messages
    it sent before $m$ have been acked.
  \item When a process receives a YCT control message, it decrements MODE
    (transitioning to normal mode if MODE becomes zero).
\end{enumerate}

Not only is Cykas conservative in enforcing FIFO, by not having a
receive-buffer, but Cykas lacks precise knowledge and will be overly
conservative in enforcing CSPS.
Cykas has a process-wide state, keeping a single MODE variable for
the whole process, incremented when an eager-send is received, and decremented
by a YCT. In the run in Figure~\ref{fig:dependencies}, processes $j$ and $k$
only send. As they never receive messages, they can network-send immediately
upon a causal-send (except if it was a message to the same destination before
the ack of the previous message, which never happens in this run).
Also, all messages network-sent to $i$ are by eager-sends (as acks are still
missing).
The timeline for process $i$ shows the value of the MODE variable, which is
increased upon receiving $b$ and $a$, and decremented by their respective
YCTs (dashed lines).

But in this run another message ($c$) from $k$ is eager-sent to $i$, and
happens to arrive before the YCT regarding $a$, preventing MODE from reaching
zero, and $m$ to be network-sent, until later, when the YCT for $c$ arrives.
This breaks optimality, as message $c$ did not happen before $m$, and
should be irrelevant for the decision to network-send $m$. Unfortunately it
causes Cykas to network-send later, being $m$ buffered in the interval shown
by the red rectangle (below the timeline for process $i$).

Worse than some unnecessary delay, this issue causes a liveness problem in
Cykas. Consider similar runs, in which at least two processes keep sending
messages to $i$, always as eager-sends (something possible if they also
alternate, sending messages to other processes, with some acks always
missing at the send time). It can happen that, before some YCT still missing
arrives, another eager-sent message arrives, preventing MODE from reaching
zero. This can go on, in an infinite run, causing starvation in the ability of
$i$ to network-send $m$, making $m$ (as well as any other message subsequently
causal-sent by $i$) remain in the send-buffer forever.

The problem lies in the lack of precise information by Cykas, collapsing in a
single process-wide MODE integer information both about relevant messages,
that happened before $m$, with information about irrelevant concurrent
messages.
Cykas also has two other weak points. One is that YCTs carry no
identification and imply the decrement of a counter, which is not an
idempotent operation. This requires exactly-once delivery of YCTs, making
Cykas unsuitable to be adapted to the assumption of an unreliable network.
The other point is the way Cykas stores and processes metadata at each
process. Both the space at each process and processing cost per message at the
sender potentially grows as $O(no)$, the product of the number of processes $n$
with the process outdegree $o$ (the number of destinations of messages sent
by the process).

\section{A new causal delivery algorithm with hybrid buffering}

\subsection{System model and assumptions}

We assume an asynchronous unreliable network. There is no global clock, no
knowledge about relative processing speeds and no bound for the time it takes
a message to arrive. Messages can be lost, duplicated, and reordered. If
messages keep being sent between two processes, eventually some message will
get through.
We use the unreliable network assumption for the following reasons: as the
results are stronger, applicable to reliable networks; to stress the
suitability of the algorithm to unreliable networks, unlike some other
algorithms; and to stress the little space overhead that causal delivery will
impose over what would have to be stored anyway (either by the
algorithm or in the transport layer) to achieve reliability, regardless of
delivery order.

Algorithms that assume a reliable network, and whose practical implementations
use TCP, will be vulnerable to network problems: if a connections fails
somehow (e.g., due to a change of IP in a mobile host), some messages already
written to the TCP stream may get lost. Also, many systems involving possibly
thousands of processes per host need to multiplex one (or a few) TCP streams
connecting hosts. This can potentially cause Head-of-Line blocking, causing
delays in causally unrelated messages between hosts, if they share a TCP
stream.
Our assumption of unreliable network will allow different implementations
(e.g., over UDP) without suffering from these problems.

The algorithm assumes that process do not fail, but makes no assumption about
the number or identities of all processes involved (i.e., no assumption about
some a priori configuration), only that each process has a globally unique
identifier. Processes may be added along time. Contrary to algorithms which
assume $n$ processes, and have $n$-dependent data structures, our algorithm
will have data structures with numbers of entries which may depend on the
indegree or outdegree of the communication graph, which the algorithm does not
need to know. For large systems with many thousands of processes, but where
each node communicates only with some tens or hundreds of processes, this is
important for scalability and dynamicity.

\subsection{Overview}

Our algorithm can be used for unicast of multicast. We will present it for
the unicast case and cover multicast in Section~\ref{sec:multicast}.
We start by a basic version, enforcing the CSPS strategy,
and will present the SPS-optimal variant in Section~\ref{sec:optimizized-algorithm},
allowing better delivery latency in some scenarios.
There are three components in the basic algorithm:
\begin{itemize}
  \item Tracking still missing acks of messages sent before a given message
    $m$ is sent, to know if a permit is needed for $m$ when network-sending it,
    and when to send the permit in that case.
  \item Tracking still missing permits regarding messages delivered before a
    given message $m$ is causal-sent, to known if it is already possible to 
    network-send $m$ or if it must still be kept in the send-buffer.
  \item Enforcing FIFO.
\end{itemize}

Regarding acks of sent messages, the essential insight is that, if we
identify sent messages by a total order (using a process wide clock assigning
integers to messages), we do not need to track the set of missing acks per
message. If there is an ack missing regarding messages before message number
$k$, it will be missing in what concerns any message with number $l > k$.
Therefore, we only need a data structure which allows knowing efficiently what
is the number of the oldest message with the ack still missing. Together with
keeping the payloads themselves until acked, for network fault tolerance, this
will prompt a \emph{sliding array} data structure: a dynamic array, which
grows being appended, but which keeps only elements in a sliding window
starting from a \emph{first} index until the \emph{next} index to which an
element will be appended.  When an ack arrives, we use the message id as index
in the sliding array, for effectively $O(1)$ lookup.

Regarding permits for received messages, a similar insight can be made. If
a permit still missing prevents message number $k$ from being
network-sent, it will also prevent message number $l > k$ from being network-sent.
(This applies to CSPS but not to SPS, as we will see.)
But permits arrive from different origins, and we cannot use their message
ids as indexes. What the algorithm does is assign an increasing permit number per
message delivered needing a corresponding permit. When a message $m$ is
causal-sent, the number to be assigned to the next permit is kept with
$m$. This allows knowing if there are still some missing permits, by comparing
the number kept with $m$ with the number assigned to the oldest delivered message
with corresponding permit still missing. This prompts a data structure
resembling a sliding array. But for computation efficiency there is a problem:
when a permit arrives we cannot use the message id as an index for lookup
in such a sliding array; we would need to traverse the elements themselves
until a match is found. To overcome this problem we use a different data
structure: a sliding map, which combines a sliding array of booleans (i.e., a
sliding bit-array) with a map which stores indexes into the sliding bit-array as
values, to allow effectively $O(1)$ lookup by key.

\subsection{The sliding array and sliding map data-structures}

A \emph{sliding array} is aimed at efficiently storing elements in
a sliding window starting from a \emph{first} index until the \emph{next}
index to which an element will be appended. The operations which we use in the
algorithm are: the update operations $\af{add(E)}$ and $\af{remove()}$
(removing and returning the first element);
the query operations $\af{first}$, $\af{next}$, $\af{size}$, and
$\af{peek}$ (which returns the oldest element for a non-empty array);
access by index, as a normal array. An implementation can keep two integer indexes
(first and next) and a very small dynamic array of pointers to larger arrays,
as shown in Figure~\ref{fig:sliding-array}.

\begin{figure*}
    \centering
    \begin{subfigure}{0.48\textwidth}
        \centering
        \includegraphics[scale=0.9]{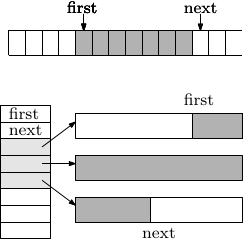}
    \caption{sliding array and a concrete representation.}
    \label{fig:sliding-array}
    \end{subfigure}%
    \hfill
    \begin{subfigure}{0.48\textwidth}
        \centering
        \includegraphics[scale=0.9]{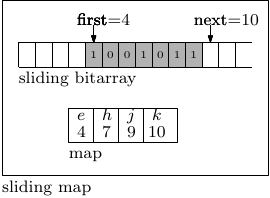}
      \caption{sliding map (abstract representation) after adding elements $a,
      b, c, d, e, f, g, h, i, j, k$, and removing elements $a, b, c, d, f, g,
      i$.}
    \label{fig:sliding-map}
    \end{subfigure}%
  \caption{Sliding array  and sliding map (combination of sliding bit-array and
  map).}
\end{figure*}

When $\af{first}$ advances past all elements in the oldest array, the array
is discarded and elements in the pointer array shifted, discarding the first
entry. When the array where elements are added to becomes full, a new array is
allocated and pointed to by a new entry in the pointer array. In practice, the
pointer array may never need to grow, if starting from some base capacity, and
be inlined in the main object to avoid an indirection if it does not need to
grow (e.g., as the Rust $\af{smallvec}$~\cite{smallvec}). Operations will be
in amortized constant time, with most of the time operations either only
accessing cached memory, or having one non-cached access.

A \emph{sliding map} also keeps track of elements associated with a
sliding window, but allows access by key, as a map, instead of by index.
The operations we use in the algorithm are: the update operations
$\af{add(E)}$, $\af{remove(E)}$, and the query operations $\af{first}$ and
$\af{next}$, which return indexes, as for the sliding array. Contrary to the
sliding array, which stores a contiguous range of elements, where only the
oldest element is removed, in the sliding map there will be positions with no
associated element (as removes are by key). An implementation of a sliding map
can be obtained by combining a sliding array of bits, where the value denotes
presence at the given index, with a map from keys to indexes in the sliding
bit-array, as illustrated in Figure~\ref{fig:sliding-map}.
When a remove happens to be for the element at the first position,
$\af{first}$ is advanced by skipping all 0 bits until a 1 is found (or the
next position is reached, if the map becomes empty). Comparison with
$\af{next}$ during the traversal can be avoided by keeping a sentinel 1 at the
$\af{next}$ position. Assuming a map implementation with amortized effectively
constant time operations (as a typical hashtable), the sliding map will also
provide amortized effectively constant time operations.

A variant, \emph{indexable sliding map}, will be used by the SPS-optimal algorithm.
It also allows access by index as a sliding array, and provides the query
operation $\af{index(E)}$ to obtain the index of an element (doing a map lookup).
A possible implementation can be using a map together with a
$\af{SlidingArray}[1 + E]$, i.e., a sliding array of optional $E$'s, instead
of a sliding bit-array.

\subsection{The basic unicast algorithm}
\label{sec:basic-algorithm}

Algorithm~\ref{alg:basic-unicast} is the basic unicast causal delivery
algorithm, enforcing CSPS + FIFO, for a generic process $i$ from the set of process
identifiers $\ids$. We use $i$ subscripted terms for process state
variables or for actions that happen at process $i$ (like send or receive).
Unsubscripted variables are local/temporary constants or variables. We use
$\gets$ for variable assignment and $=$ for let binding.
For presentation simplicity we use a single Msg record both for messages in the
send-buffer, in the unacked buffer and in transit (the $\field{rcv}$ field
is not needed for messages in transit).
The $\field{per}$ field is either an integer index (while in the
send-buffer) or a boolean (denoting a permit is needed, elsewhere). The Per
record is used for permits and the Rcv record for messages received but not yet
delivered (as FIFO needs to be enforced).
All maps are assumed to only keep non-bottom values; a non-mapped
key implicitly maps to bottom, e.g., maps to zero for non-negative integer
values.

\begin{algorithm2e*}

\caption{Basic causal delivery algorithm for process $i$.}
\label{alg:basic-unicast}
\newcommand\f[1]{\field{#1}}
\begin{multicols}{2}
  \types{
    $\ids$, process identifiers \;
    $\P$, message payloads \;
    $\Msg: \record$ \{ \;
    \quad $\f{rcv} : \ids$, receiver id \;
    \quad $\f{mid} : \nat$, message id \;
    \quad $\f{pid} : \nat$, predecessor from same sender \;
    \quad $\f{per} : \nat + \bool$, permit index / flag \;
    \quad $\f{pl} : \P_\bot$, message payload \;
    $\}$, message record \;
    $\Per: \record$ \{ \;
    \quad $\f{snd} : \ids$, sender id \;
    \quad $\f{mid} : \nat$, message id \;
    \}, permit record \;
    $\Rcv: \record$ \{ \;
    \quad $\f{mid} : \nat$, message id \;
    \quad $\f{pl} : \P$, message payload \;
    \quad $\f{per} : \bool$, permit flag \;
    \}, receive buffer record \;
  }
  \BlankLine
  \state{
    $ck_i : \nat = 0$, clock for message ids \;
    $u_i : \af{SlidingArray}[\Msg]$, unacked messages \;
    $p_i : \af{SlidingMap}[\Per]$, missing permits \;
    $ls_i: \ids \pfunc \nat = \emptyset$, last sent message ids \;
    $ld_i: \ids \pfunc \nat = \emptyset$, last delivered message ids \;
    $sb_i : \Msg^*$ = [], send buffer \;
    $rb_i : \ids \pfunc \nat \pfunc \Rcv = \emptyset$, receive buffer \;
  }
  \BlankLine

  \procedure({$\trySend()$}) {
    \While {$sb_i \neq []$} {
      $m = sb_i.\peek$ \;
      \If { 
        $p_i.\first < m.\f{per}$
      }
      {
        \return
      }
      $sb_i.\remove()$ \;
      $m.\f{per} \gets u_i.\size > 0$ \;
      $u_i.\add(m)$ \;
      $\send_i(m.\f{rcv}, \atom{msg}, m)$ \;
    }
  }

  \BlankLine
  \on({$\causalSend_i(j, p)$}) {
    $m = \Msg \{ \f{rcv} : j, \f{mid} : ck_i, \f{pid} : ls_i[j],$ \;
    $\hphantom{m = \Msg \{}
    \f{per} : p_i.\next, \f{pl} : p \}$ \;
    $ls_i[j] \gets ck_i$ \;
    $ck_i \gets ck_i + 1$ \;
    $sb_i.\add(m)$ \;
    $\trySend()$ \;
  }

  \BlankLine
  \on({$\receive_i(j, \atom{msg}, m)$}) {
    \If {$m.\f{mid} \leq ld_i[j]$} {
      $\send_i(j, \atom{ack}, b.\f{mid})$ \;
      \return
    }
    $e = rb_i[j]$ \;
    $e[m.\f{pid}] \gets \Rcv(m)$ \;
    \While {$ld_i[j] \in \dom(e)$} {
      $b = e.\remove(ld_i[j])$ \;
      $ld_i[j] \gets b.\f{mid}$ \;
      \If {$b.\f{per}$} {
        $p_i.\add(\Per\{ \f{snd} : j, \f{mid} : b.\f{mid} \})$ \;
      }
      $\send_i(j, \atom{ack}, b.\f{mid})$ \;
      $\deliver_i(b.\f{pl})$ \;
    }
  }

  \BlankLine
  \on({$\receive_i(j, \atom{ack}, n)$}) {
    \If {$n < u_i.\first$} {
      $\send_i(j, \atom{permit}, n)$ \;
      \return
    }
    $u_i[n].\f{pl} \gets \bot$ \;
    \If {$n = u_i.\first$} {
      $u_i.\remove()$ \;
      \While {$u_i.\size > 0$} {
        $m = u_i.\peek$ \;
        \If {$m.\f{per}$} {
          $\send_i(m.\f{rcv}, \atom{permit}, m.\f{mid})$ \;
        }
        \If {$m.\f{pl} \neq \bot$} {
          \return
        }
        $u_i.\remove()$ \;
      }
    }
  }

  \BlankLine
  \on({$\receive_i(j, \atom{permit}, n)$}) {
    $p_i.\remove(\Per\{\f{snd}: j, \f{mid}: n\})$ \;
    $\trySend()$ \;
  }

  \BlankLine
  \periodically() {
    \For {$m$ \kwin $u_i$} {
      \If {$m.\f{pl} \neq \bot$} {
        $\send_i(m.\f{rcv}, \atom{msg}, m)$ \;
      }
    }
    \For {$p$ \kwin $p_i$} {
      $\send_i(p.\f{snd}, \atom{ack}, p.\f{mid})$ \;
    }
  }

\end{multicols}
\end{algorithm2e*}

The common way to enforce FIFO is to have a sequence number per destination,
but we use messages ids from a single sequence, regardless of destination.
We use the following approach: the id of the predecessor from the same origin
is sent with the message (field $\field{pid}$). For this we keep the last sent
message per destination ($ls_i$ map, line 24). Also, the receive buffer
($rb_i$, line 27) is designed as a map from senders to a map from predecessor
ids to Msg records. This allows $O(1)$ processing of received messages even if
they arrive out of order: an entry is added to the map for the given sender,
with predecessor id as key (line 49, where the Rcv constructor extracts the
relevant fields). After a message is received, it can be delivered if the
last delivered ($ld_i$ map, line 25) message from that sender becomes present
as key.  When delivering, the map entry is removed, $ld_i$ is updated and the
test repeated, looping until a message is missing (line 50, where $\af{dom}$
represents map domain, i.e., the map keys).

When a delivered message was flagged as needing a corresponding permit
(line 53), a permit is added (line 54) to the missing permits sliding map
(line 23), the Per record having the sender and message id. A 
causal-sent (line 37) message will be associated (line 39) to the
$\af{next()}$ index of the permits sliding map, before being added (line 42) to the
send-buffer (line 26). This will delay (lines 31, 32) the network-send (line
36) until all permits before that index (those for each message previously
delivered needing a permit) are received, upon which $p_i.\af{first}()$ reaches
$m.\field{per}$.

The $\af{trySend}()$ procedure (lines 28--36) tries to send messages from the
send-buffer, in the order they were buffered, if conditions allow. It is
invoked after a message is added to the send-buffer (line 43) or when a permit
is received and removed from the permits sliding map (lines 71--73).  It
checks if all needed permits for the first message in the buffer have been
received (condition at line 31 being false), upon which the message is removed
from the buffer, flagged as needing a permit if there is any unacked message
(line 34), added (line 35) to the unacked messages sliding array (line 22) and
network-sent. The procedure is repeated for the remaining buffer upon a
successful send.

When a message is delivered, it is acked (lines 55, 56). Upon receiving an
ack (line 57), the payload is removed from the entry in the unacked
sliding array, using the message id as index (line 61). If the ack corresponds
to the first entry in the sliding array (line 62), then the contiguous range
of messages already acked, starting from the one just acked, is removed from
the sliding array. For each of these messages, and also for the still unacked
message that becomes the new first position, if the message had been flagged
as needing a permit, the permit is sent (lines 66, 67).

This algorithm tolerates network faults (message loss, duplication,
reordering). In general, the effect of all receives, whether $\atom{msg}$,
$\atom{ack}$, or $\atom{permit}$ is designed to either ignore the message but
possibly replying as usual, or having an idempotent effect, in case there is
no local knowledge to decide whether it is a duplicate.

Any message received is tested as a possible duplicate already
delivered (line 45), but acked anyway even if so, as a previous ack may have
been lost. If it has already been received but not yet delivered, an
idempotent action will be performed (line 49).
Acks are treated in a similar way. When an ack is received of a message
no longer present in the unacked buffer, it is ignored, but a permit sent
anyway (lines 58--59), as a previous one may have been lost. It may happen
that no such permit was sent but no information exists now to know if that is
the case.  Sending it will be harmless, as it will cause the removal of a
non-existing entry in the sliding map (line 72).

If messages are lost, after some time with no progress retransmission occurs,
as summarized by lines 74--79. (An actual implementation can use more refined
timeouts, e.g., per message.) Periodically, each message in the unacked
sliding array for which an ack has not been received will be sent (lines
75--77). Also, for each missing permit, an ack will be sent (lines 78--79), as
if the message had just been delivered. This will cause the sender, upon
receiving the ack, to (re)transmit the missing permit, when the message is no
longer in the unacked buffer.
This mechanism is needed, because a permit may be lost and the sender has
removed the corresponding message from the unacked buffer, in which case
it will never retransmit the message.

\subsection{Correctness of the basic algorithm}

\subsubsection{Safety}
From Proposition~\ref{prop:fifo-csps-causal}, it is only necessary to show
that it enforces FIFO and CSPS.

\begin{proposition}
The basic algorithm enforces FIFO.
\end{proposition}
\begin{proof}
Each message $m$ carries the id of the predecessor message $p$ from the same
sender, and the algorithm only delivers message $m$ when the last delivered
one from that sender is $p$ (when condition at line 50 is true).
\end{proof}

\begin{proposition}
The basic algorithm enforces CSPS.
\end{proposition}
\begin{proof}
1) If there are missing acks of messages previously sent by a process
when the process network-sends $m$, then $m$ will be flagged as
needing a subsequent permit (line 34). This will make a corresponding entry
be created in the missing permits sliding-map when $m$ is delivered (lines
53--54).
2) A permit for message with id $n$ is only sent when the index of the first
message in the unacked sliding-array is not less than $n$ (conditions at lines
58 and 62), which implies that all messages by the same process that happened
before (with id less than $n$) have been acked, and therefore, delivered, as
acks are only sent for messages delivered (lines 46, 55, and 79). 
3) Process $i$ only network-sends message $m$ after receiving all permits for
messages that have been delivered before $m$ is causal-sent
(the permits that have an index less than $m.\field{per}$)
delaying the send while that is not true (lines 31--32).
4) Points 1, 2, and 3 together imply that $m$ is network-sent by $i$ only
when all messages that happened before $m$, sent by any process that sent
messages to $i$, have been delivered.
\end{proof}

\subsubsection{Liveness}

The essential point for liveness is that, for any causal-sent
message that has been buffered, the network-send cannot be postponed
indefinitely, due to the delivery of messages after the causal-send: it
depends only on permits concerning previously delivered messages.

\begin{lemma}
\label{lemma:unacked-buffer}
If a message is added to the unacked buffer, it will be eventually received and
delivered at the destination, and removed from the unacked buffer.
\end{lemma}
\begin{proof}
If a message is added to the unacked buffer, the entry remains there until
all previously sent messages, and the message itself, have been acked (lines
62--70). While a message is still unacked, the message is periodically
retransmitted (lines 75--77) until an ack is received. For each destination,
the oldest such message still not delivered will be eventually received,
delivered and acked (lines 50--56). An ack will be eventually received, even
if acks are lost, because the ack is sent anyway even if the message has
already been delivered (lines 45--47). This implies that a message put in the
unacked buffer will be eventually received and delivered at the destination,
and removed from the buffer.
\end{proof}

\begin{lemma}
  \label{lemma:permits}
Each missing permit will be eventually received and removed from the
sliding-map.
\end{lemma}
\begin{proof}
An entry in the permits sliding-map is created when the corresponding
message $m$, that was sent flagged as needing a permit and put in the unacked
buffer, is delivered.
By Lemma~\ref{lemma:unacked-buffer}, eventually, at the sender, the entry
corresponding to $m$ will be removed from the unacked buffer.
For each still missing permit, an ack will be periodically sent to the
corresponding sender (lines 78--79). Eventually, the sender will
always reply to such acks with a permit (lines 57--59). Eventually, a permit is
received, causing the removal of the corresponding entry from the permits
sliding-map (lines 71--72).
\end{proof}

\begin{proposition}
For the basic algorithm, any causal-sent message will be eventually delivered.
\end{proposition}
\begin{proof}
Any message $m$, causal-sent by a process $i$, will remain in the send-buffer
only if any permit corresponding to previously delivered messages is still
missing (lines 31--32).
By Lemma~\ref{lemma:permits}, all permits, namely those corresponding to
messages delivered at process $i$ before $m$ is causal-sent will be received and
removed from the sliding-map, upon which $m$ will be
added to the unacked buffer (line 35), and by
Lemma~\ref{lemma:unacked-buffer}, $m$ will be eventually delivered.
\end{proof}

\subsection{An SPS-optimal algorithm}
\label{sec:optimizized-algorithm}

The basic CSPS enforcement algorithm already allows scalable throughput
but, being conservative in the requirements for when to network-send messages,
it will result in some added delivery latency compared with the optimal
allowed by the SPS strategy.  Now we show how to enforce SPS optimally, to
allow better delivery latency in some scenarios.
The basic algorithm is conservative in three ways, that can be improved upon:
\begin{enumerate}
  \item From SPS, the need for a permit when network-sending to a process $j$
    is only if acks of messages to processes other than $j$ are still missing.
    The basic algorithm, enforcing CSPS, flags as needing a permit if there is
    any unacked message, even for the same receiver.
  \item If a message to a process $j$ is flagged as needing a permit, for the
    same reason, the basic algorithm only sends the permit when there is no
    ack for previously sent messages missing, but SPS only requires that there
    is no missing ack of messages to processes other than $j$.
  \item The basic algorithm delays a network-send of a message $m$ to a
    process $j$ if there is any missing permit for messages delivered before
    $m$ being causal-sent, even for messages from $j$. SPS only requires
    delaying the network-send if there are missing permits from processes
    other than $j$.
\end{enumerate}

Improving these three points does not allow the simple and elegant approach of
the basic algorithm. It requires some care and modifications to data
structures and algorithm for it to remain with amortized constant time
computation complexity.

A consequence of the third point is that, when a process causal-sends
$m_1$ and later $m_2$, to different destinations, it may be possible to
network-send $m_2$, when the permits missing for $m_2$ are all
from its destination, but not yet $m_1$, as some of those permits may be
common to $m_1$. This implies that the SPS-optimal algorithm cannot
assume that the network-send order coincides with the causal-send order, and
cannot simply use a FIFO send-buffer and look only at the first message.
This poses the problem of how to avoid having to look at the whole send-buffer
when a permit arrives.
Also, the causal-send order must be respected for the purpose of assigning
message ids, and to decide whether to flag a message as needing a permit or
when to send it: A previous causal-sent but not yet network-sent message must
be assigned a smaller id and considered as missing an ack.

The solution is to use a unified buffer, for both unacked and unsent messages,
based on a sliding-array, where each entry has an extra $\field{sent}$ field,
to denote whether the message has already been network-sent (or is still
buffered), together with two index variables to demarcate different parts of
the buffer. Figure~\ref{fig:unified-buffer-and-permits} depicts the unified
buffer after messages $a, b, \ldots n$ have been causal-sent, when $c$ is the
first unacked message. All messages before index $m1$ have been network-sent,
with some already acked (represented with $\bot$), corresponding to the
unacked buffer in the basic algorithm.
Index $m1$ corresponds to the first buffered message not yet network-sent
(the start of the range corresponding to the send-buffer). But now some
subsequent messages may have already been network-sent, and even acked.

\begin{figure}
\begin{center}
  \includegraphics[scale=1]{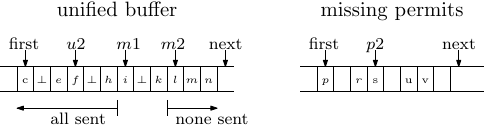}
\end{center}
  \caption{Example showing unified buffer and missing permits. All messages
  before index $m1$ have been network-sent and no message starting from index
  $m2$ have been network-sent. Some messages between $m1$ and $m2$ may have
  been network-sent and even acked ($\bot$ payload). At index $u2$ is
  the first unacked message to a different receiver than the oldest message in
  the buffer ($f.\field{rcv} \neq c.\field{rcv} = e.\field{rcv} $). Permit at
  index $p2$ is the first from a different sender ($s.\field{snd} \neq
  p.\field{snd} = r.\field{snd}$). Message at index $m2$ is the first
  depending on the permit at $p2$ ($l.\field{perm} > p2$ and $k.\field{perm}
  \leq p2$).}
  \label{fig:unified-buffer-and-permits}
\end{figure}

To avoid traversing the whole buffer starting from $m1$ whenever a permit
arrives, to decide which messages can be network-sent, we can use the
following insight: Considering the missing permits
indexable sliding map, if we maintain the index $p2$ of the first permit from
a different sender than the first permit, then no causal-sent message that
depends on it (i.e., with the $\field{perm}$ field greater than $p2$) can be
network-sent. Index $m2$ corresponds to the first such message that
cannot be network-sent due to depending on permits from different senders,
i.e., such that $u[m2].\field{perm} > p2$. Messages depending on permits with
smaller indexes can be network-sent if being sent to the sender of the first
permit.

Similarly, to avoid traversing the whole buffer up to $m2$ whenever an ack
arrives, to decide whether to send permits, we maintain another index
variable, $u2$, denoting the first unacked message to a different receiver
than the one of the first message in the buffer. No permit for messages with
id greater than $u2$ can be sent. Permits for messages with id smaller than
$u2$ can be sent if for the same receiver as the one of the first message in
the buffer.

Keeping these four variables, updating them incrementally, and using them to
decide if more messages or permits can be sent, plays an essential role in
achieving amortized constant time computation.
Whenever a permit arrives it will only trigger a traversal if it is the first
permit or the one at $p2$. Whenever an ack arrives, it will only trigger a
traversal if it is the oldest missing or for $u2$. Whenever a traversal is
performed, index variables are updated and no range of the unified buffer or
permits map is traversed more than twice. This results in amortized constant
time per operation.

Algorithm~\ref{alg:SPS-optimal-unicast} contains the SPS-optimal causal
delivery algorithm. We omit the types section, as only the Msg record differs
from the basic algorithm, having the extra $\field{sent}$ field. The state has
now the unified buffer, the missing permits is now an indexable sliding map,
and there are those four new index variables ($m1$, $m2$, $p2$, and $u2$).

\begin{algorithm2e*}

\caption{SPS-optimal causal delivery algorithm for process $i$.}
\label{alg:SPS-optimal-unicast}
\newcommand\f[1]{\field{#1}}
\vspace*{-6mm}
\begin{multicols}{2}
  \state{
    $ck_i : \nat = 0$, clock for message ids \;
    $u_i : \af{SlidingArray}[\Msg]$, unified buffer \;
    $p_i : \af{IdxSlidingMap}[\Per]$, missing permits \;
    $ls_i: \ids \pfunc \nat = \emptyset$, last sent message ids \;
    $ld_i: \ids \pfunc \nat = \emptyset$, last delivered message ids \;
    $rb_i : \ids \pfunc \nat \pfunc \Rcv = \emptyset$, receive buffer \;
    $p2_i: \nat = 0$, 1$^{\text{st}}$ permit from another sender \;
    $m1_i: \nat = 0$, 1$^{\text{st}}$ buffered msg \;
    $m2_i: \nat = 0$, 1$^{\text{st}}$ msg requiring $p2_i$ \;
    $u2_i: \nat = 0$, 1$^{\text{st}}$ msg needs ack another rcv\;
  }
  \BlankLine

  \on({$\causalSend_i(j, p)$}) {
    $m = \Msg \{ \f{rcv} : j, \f{mid} : ck_i, \f{pid} : ls_i[j],$ \;
    $\hphantom{m = \Msg \{}
    \f{per} : p_i.\next, \f{pl} : p, \f{sent} : \false \}$ \;
    $ls_i[j] \gets ck_i$ \;
    $ck_i \gets ck_i + 1$ \;
    $u_i.\add(m)$ \;
    $\sendInterval(u_i.\next -1 , p2_i)$ \;
  }
  \BlankLine

  \periodically() {
    \For {$m$ \kwin $u_i$} {
      \If {$m.\f{sent} \land m.\f{pl} \neq \bot$} {
        $\send_i(m.\f{rcv}, \atom{msg}, m)$ \;
      }
    }
    \For {$p$ \kwin $p_i$} {
      $\send_i(p.\f{snd}, \atom{ack}, p.\f{mid})$ \;
    }
  }
  \BlankLine

  \procedure({$\sendInterval(k, p)$}) {
    $p1 = p_i.\first$ \;
    $s = p_i[p1].\f{snd}$ \hfill // undefined if no permits \;
    $r = u_i.\peek.\f{rcv}$ \hfill // undefined if no msgs \;
    \While { $k < u_i.\next \land \, u_i[k].\f{per} \leq p$ } {
      $m = u_i[k]$ \;
      \If { $\lnot m.\f{sent} \land (m.\f{per} \leq p1 \lor m.\f{rcv} = s)$ } {
        $m.\f{sent} \gets \true$ \;
        $m.\f{per} \gets u2_i < k \lor m.\f{rcv} \neq r$ \;
        $\send_i(m.\f{rcv}, \atom{msg}, m)$ \;
      }
      $k \gets k + 1$ \;
    }
    \While {$m1_i < u_i.\next \land \, u_i[m1_i].\f{sent}$} {
      $m1_i \gets m1_i + 1$ \;
    }
  }
  \BlankLine

  \procedure({$\af{updateP2M2}()$}) {
    $p1 = p_i.\first$ \;
    \If {$p1 = p_i.\next$} {
      \return
    }
    $s = p_i[p1].\f{snd}$ \;
    \While {$p2_i < p_i.\next \land {}$ \; 
    \hphantom{While } $(p_i[p2_i] = \bot \lor p_i[p2_i].\f{snd} = s) $} {
          $p2_i \gets p2_i + 1$ \;
    }
    \While {$m2_i < u_i.\next \land \, u_i[m2_i].\f{perm} \leq p2_i$} {
      $m2_i \gets m2_i + 1$ \;
    }
  }
  \BlankLine

  \on({$\receive_i(j, \atom{ack}, n)$}) {
    \If {$n < u_i.\first$} {
      $\send_i(j, \atom{permit}, n)$ \;
      \return
    }
    $u_i[n].\f{pl} \gets \bot$ \;
    \If {$n = u_i.\first$} {
      $u_i.\remove()$ \;
      \While {$u_i.\size > 0$} {
        $m = u_i.\peek$ \;
        \If {$m.\f{sent} \land m.\f{per} \land m.\f{rcv} \neq j$} {
          $\send_i(m.\f{rcv}, \atom{permit}, m.\f{mid})$ \;
        }
        \If {$m.\f{pl} \neq \bot$} {
          \kwbreak
        }
        $u_i.\remove()$ \;
      }
      $n \gets u_i.\first$ \;
    }
    \If {$n = u2_i < u_i.\next$} {
      $r = u_i.\peek.\f{rcv}$ \;
      \While {$u2_i < u_i.\next$} {
        $m = u_i[u2_i]$ \;
        \If {$m.\f{pl} \neq \bot \land m.\f{rcv} \neq r$} {
          \kwbreak
        }
        \If {$m.\f{sent} \land m.\f{per} \land m.\f{rcv} = r$} {
          $\send_i(m.\f{rcv}, \atom{permit}, m.\f{mid})$ \;
        }
        $u2_i \gets u2_i + 1$ \;
      }
    }
  }
  \BlankLine

  \on({$\receive_i(j, \atom{msg}, m)$}) {
    \If {$m.\f{mid} \leq ld_i[j]$} {
      $\send_i(j, \atom{ack}, b.\f{mid})$ \;
      \return
    }
    $e = rb_i[j]$ \;
    $e[m.\f{pid}] \gets \Rcv(m)$ \;
    \While {$ld_i[j] \in \dom(e)$} {
      $r = e.\remove(ld_i[j])$ \;
      $ld_i[j] \gets r.\f{mid}$ \;
      \If {$r.\f{per}$} {
        $k = p_i.\next$ \;
        $p_i.\add(\Per\{ \f{snd} : j, \f{mid} : r.\f{mid} \})$ \;
        \If {$p2_i = k \land j = p_i.\peek.\f{snd}$} {
            $p2_i \gets k + 1$ \;
        }
      }
      $\send_i(j, \atom{ack}, r.\f{mid})$ \;
      $\deliver_i(r.\f{pl})$ \;
    }
  }
  \BlankLine

  \on({$\receive_i(j, \atom{permit}, n)$}) {
    $k \gets p_i.\af{index}(\Per\{\f{snd}: j, \f{mid}: n\})$ \;
    \If {$k = \bot$} {
      \return 
    }
    $p1 = p_i.\first$ \;
    $p_i.\remove(\Per\{\f{snd}: j, \f{mid}: n\})$ \;
    \If {$k = p1$} {
      $k \gets p_i.\first$ \;
      $\sendInterval(m1_i, k)$ \;
    }
    \If {$k = p2_i$} {
      $l = m2_i$ \;
      $\af{updateP2M2}()$ \;
      $\sendInterval(l, p2_i)$ \;
    }
  }
  \BlankLine

\end{multicols}
\end{algorithm2e*}

The procedure $\af{sendInterval}(k, p)$ (lines 25--37), always invoked with
$p\leq p2$, attempts to network-send messages in the unified buffer starting
from index $k$, and stopping when a message depending on a permit at index
$p$ is found, i.e., with the $\af{perm}$ field greater than $p$. A buffered
message can be network-sent if it does not depend on the first missing permit
or is to the sender of that first permit (line 31). When network-sent, a
message is flagged as needing a subsequent permit if there are pending acks of
previous messages to different destinations or to a different receiver (line
33).
The procedure also advances $m1$ if needed (lines 36--37), in case it is
invoked with the interval starting at $m1$.
This procedure is invoked when a message is causal-sent (line 18) for a
singleton interval, and may be invoked when a permit is received (lines 96
or/and 100).

Procedure $\af{updateP2M2}$ (lines 38--47), invoked when the permit at $p2$
ceases to be the first from a second sender, advances the $p2$ and $m2$
state variables so that they regain their meaning.

When a permit is received (lines 88--100), if it is still
present, it is removed from the map (line 93). Only two cases makes
network-sending possible: 1) if the received permit is the first permit, the
interval of buffered messages from $m1$ up to a message depending on what
becomes the new first permit is checked for network-sending (lines 94--96); 2)
if it corresponds to the permit at $p2$, the interval of buffered messages
from $m2$ up to a message depending on the updated $p2$ is checked for
network-sending (lines 97-100).
In case 1, when receiving and removing the first permit, it may happen that
the permit at $p2$ becomes the new first permit, and the same procedure as for
case 2 needs to be performed. In this case condition at line 97 will also
become true, $p2$ and $m2$, will be updated, and the second $\af{sendInterval}$
(line 100) will be invoked.

Receiving a message (lines 72--87) is similar to the basic algorithm. The only
difference is checking the need to update $p2$ when a permit is added: if
there was no missing permit from a second sender (i.e, $p2_i = p_i.\af{next}$)
and the message is from the sender of first permit (including the case when
the just added permit becomes the first), then $p2$ is advanced to the new
next (lines 84--85).

Handling an ack has some similarity to handling a permit, but now it attempts
to send permits, if possible. Again, only two cases matter for the ability to
send new permits:
1) if the received ack is the first in the buffer, permits will be sent for
messages network-sent to a different receiver flagged as needing a permit
(lines 57--58);
messages already acked will be removed until a message missing an ack is
found (lines 59--61);
2) if it corresponds to the ack at $u2$, the buffer is traversed from $u2$ up
to an unacked message to a different receiver than the first unacked message
(lines 67--68), which will become the new $u2$. A permit will be sent for each
network-sent message to that same receiver that needs a permit (lines 69--70).
In case 1, when receiving the ack for the first message and removing
contiguous acked messages, it may happen that $u2$ becomes the new first
message, and the same procedure as for case 2 needs to be performed.
In this case condition at line 63 will also become true, and the
corresponding code executed.

The algorithm correctness follows from the same arguments as the basic
algorithm, and the proof, which we omit, has the same structure.
For presentation simplicity, we have presented an algorithm which is
SPS-optimal only in reliable networks. It ensures safety and liveness in
unreliable networks, but under message loss, namely if a permit message is
lost, a subsequent ack retransmission may not trigger a permit retransmission
in optimal time. This is irrelevant in practice as message loss is typically
very rare. It would be easy to make handling retransmissions optimal, but we
have refrained to do so to avoid cluttering an already intricate algorithm
with such non-essential details.

\subsection{Causal multicast}
\label{sec:multicast}

We now describe briefly a multicast version of the basic algorithm. For
presentation simplicity it lacks a few optimizations. We leave out a multicast
variant of the SPS-optimal algorithm, which can be obtained by a similar
reasoning.

A multicast is useful and cannot be simulated by using a unicast algorithm and
performing a multicast by sending a sequence of unicast messages. To
understand why, consider multicasting $m$ to two receivers $j$ and $k$.
Simulating it by sending two unicast messages, $m_1$ to $j$ and $m_2$ to $k$,
with the same content, will make a subsequent send $m_3$ by $j$ to $k$, after
delivering $m_1$, not causally dependent on $m_2$, allowing $m_3$ to be
delivered at $k$ before $m_2$. This is a wrong behavior if the idea is to
simulate the multicast of a single message $m$, in which case any message sent
after delivering $m$ must be delivered after $m$ at any common destination.

Algorithm~\ref{alg:multicast-basic} presents the differences from the basic
algorithm. The destination of a message is now a set of processes, and the
record Msg contains now an extra field ($\field{unack}$), not used for
messages in transit, which stores process ids for which acks are still missing.
A causal-send stores the process clock in the last-sent field, for all
destinations (lines 19--20). A network-send will also loop over the set of
receivers (lines 15--16) for the first time, or the set of receivers with acks
still missing (lines 42--43), if retransmitting.

\begin{algorithm2e*}[t]

\caption{Basic multicast causal delivery algorithm (differences from the basic algorithm).}
\label{alg:multicast-basic}
\newcommand\f[1]{\field{#1}}
\begin{multicols}{2}
  \types{
    $\Msg: \record$ \{ \;
    \quad $\f{rcv} : \pow{\ids}$, set of receivers \;
    \quad $\f{unack} : \pow{\ids}$, unacked receivers \;
    \quad \ldots \;
    $\}$, message record \;
  }
  \BlankLine

  \procedure({$\trySend()$}) {
    \While {$sb_i \neq []$} {
      $m = sb_i.\peek$ \;
      \If { 
        $p_i.\first < m.\f{per}$
      }
      {
        \return
      }
      $sb_i.\remove()$ \;
      $m.\f{per} \gets u_i.\size > 0 \lor \setsize{m.\f{rcv}} > 1$ \;
      $u_i.\add(m)$ \;
      \For {$j$ \kwin $m.\f{rcv}$} {
        $\send_i(j, \atom{msg}, m)$ \;
      }
    }
  }

  \BlankLine
  \on({$\causalSend_i(J, p)$}) {
    $m = \Msg \{ \f{rcv} : J, \f{unack} : J, \ldots \}$ \;
    \For {$j$ \kwin $J$} {
      $ls_i[j] \gets ck_i$ \;
    }
    $ck_i \gets ck_i + 1$ \;
    $sb_i.\add(m)$ \;
    $\trySend()$ \;
  }

  \BlankLine
  \on({$\receive_i(j, \atom{ack}, n)$}) {
    \If {$n < u_i.\first$} {
      $\send_i(j, \atom{permit}, n)$ \;
      \return
    }
    $u_i[n].\f{unack}.\af{remove}(j)$ \;
    \If {$u_i[n].\f{unack} = \emptyset$} {
      $u_i[n].\f{pl} \gets \bot$ \;
    }
    \If {$n = u_i.\first$} {
      \While {$u_i.\size > 0$} {
        $m = u_i.\peek$ \;
        \If {$m.\f{pl} \neq \bot$} {
          \return
        }
        \If {$m.\f{per}$} {
          \For {$j$ \kwin $m.\f{rcv}$} {
            $\send_i(j, \atom{permit}, m.\f{mid})$ \;
          }
        }
        $u_i.\remove()$ \;
      }
    }
  }

  \BlankLine
  \periodically() {
    \For {$m$ \kwin $u_i$} {
      \For {$j$ \kwin $m.\f{unack}$} {
        $\send_i(j, \atom{msg}, m)$ \;
      }
    }
    \For {$p$ \kwin $p_i$} {
      $\send_i(p.\f{snd}, \atom{ack}, p.\f{mid})$ \;
    }
  }

\end{multicols}
\end{algorithm2e*}

To enforce CSPS, a message $m$ multicast to more than one receiver will be
included in the messages that require knowledge about having been delivered.
This implies that they will require a permit, even if network-sent with no
acks from previously sent messages missing (line 13). It also means that, to
send the permit to the set of receivers, not only all acks of messages with
smaller message ids need to have been received, but also the acks of the
message itself.
(An optimization not shown in the algorithm would be being able to send the
permit when a single ack of the message itself is missing, to the corresponding
process.) 

This implies changing how acks are processed: receiving an ack
removes the sender from the $\field{unack}$ set (line 28); when it becomes
empty the payload is removed (lines 29--30); if the ack is for the first
message missing acks, the receivers for each message in a contiguous range
for which no ack is now missing will be sent the corresponding permit (lines
31--39).

\section{Future work}
\label{sec:future}

After designing the new approach and algorithm, an obvious future work is
performance evaluation.
This algorithm is very different from classic receiver-buffering approaches,
and there is no much point in it being used, or evaluated, at small scale. For
large scale systems, it obviously wins in message metadata and computation
overhead, and so the question is how delivery latency is impacted. But this impact
depends on the usage scenarios: communication patterns and how the time
handling messages compares with network latency. A proper evaluation will
depend on collecting and characterizing large scale realistic scenarios (e.g.,
involving message-based microservices).
Purely synthetic scenarios will produce results according to how parameters are
selected, and may be uninteresting, unless the parameters reflect realistic
situations. Either way, a proper evaluation is beyond the current work, and
will require a full paper,
similarly to how the classic KS~\cite{DBLP:journals/dc/KshemkalyaniS98}
algorithm was evaluated by~\textcite{DBLP:journals/tpds/ChandraGK04}.

Being the first hybrid-buffering approach, it opens the way for future
research on other such algorithms. The obvious motivation will be reducing
delivery latency. It will be interesting to see whether hybrid-buffering
algorithms can be devised which use receiver-buffering beyond just enforcing
FIFO, to allow possibly earlier sending, compared to what SPS allows, at the
cost of sending some metadata about causal predecessors. Such a
tradeoff could possibly be done at runtime, per message: e.g., sending some
metadata to control receiver-buffering, if the size is not prohibitive, but
falling back to delaying the network-send by SPS, otherwise.

\section{Conclusion}
\label{sec:conclusion}

Traditional topology-agnostic causal delivery algorithms are almost always
based on buffering at the receiver before delivering. They are not suitable
for large scale systems with thousands of processes, exhibiting a prohibitive
cost in the size of metadata. Sender-buffering algorithms have not been much
pursued, given the poor latency and throughput performance of the classic
sender-buffering algorithm by ~\textcite{DBLP:conf/dagstuhl/MatternF94}, until
a recent algorithm, Cykas~\cite{TongKuper2024}, which decouples delivery
from the ability to network-send subsequent messages.

In this paper, first we have introduced criteria for comparing causal
delivery algorithms that are not latency optimal and performed a
comparison of representative topology-agnostic algorithms. We have noticed
that throughput scalability is an important criteria to distinguish such
algorithms, and shown that no purely sender-buffering algorithm can achieve
scalable throughput, which can only be achieved by either receiver- or
hybrid-buffering.

Then, we have introduced the Sender Permission to Send (SPS) enforcement
strategy and a novel approach to achieve causal delivery, by enforcing FIFO
and SPS simultaneously, which we have shown to imply causal delivery.

We have presented the first hybrid-buffering causal delivery algorithm, in
three variants, based on the novel SPS+FIFO approach, using sender-buffering
to enforce SPS and receiver-buffering to enforce FIFO. It achieves optimal
metadata complexity ($O(1)$ integers/identifiers) for messages in transit.
Moreover, by a careful design, we have achieved optimal computation
complexity, i.e., amortized effectively constant time per message. As fas as
we know, it is the first topology-agnostic throughput-scalable algorithm to
achieve these properties.

The new algorithm is not latency optimal, but it is the only topology-agnostic
algorithm suitable for large scale systems, where the metadata cost
of classic algorithms is prohibitive, or in scenarios where latency is less
important than throughput and the overhead of processing each message should be
minimized. 

The hybrid-buffering approach opens the way for future algorithms with
different tradeoffs, improving latency by using receiver-buffering beyond
just enforcing FIFO, at the cost of sending some metadata about causal
predecessors.

\section*{Acknowledgments}

This work is financed by National Funds through the Portuguese funding agency,
FCT - Fundação para a Ciência e a Tecnologia, within project UID/50014/2025.
DOI: \href{https://doi.org/10.54499/UID/50014/2025}{10.54499/UID/50014/2025}

\printbibliography

\end{document}